\documentclass[11pt]{article}

\usepackage[linesnumbered,vlined,ruled]{algorithm2e}
\usepackage{graphicx, amsmath, amssymb, amsthm}
\usepackage{a4wide}
\usepackage{amsfonts}
\usepackage{authblk}

\newcommand{\tw}{\operatorname{tw}}
\newcommand{\sm}{\setminus}
\newcommand{\cO}{{\cal O}}
\newcommand{\cC}{{\cal C}}
\newcommand{\cP}{{\cal P}}

\newcommand{\cF}{{\cal F}}

\newcommand{\cOs}{{\cal O}^*}
\newcommand{\pmc}{potential maximal clique}

\newcommand{\npmc}{1.7347}
\newcommand{\fd}{{\sc Minimum $\mathcal{F}$-Deletion}}
\newcommand{\mitwt}{{\sc Maximum Induced Subgraph of Treewidth $
\leq t$}}
\newcommand{\mis}{{\sc Maximum Independent Set}}
\newcommand{\mim}{{\sc Maximum Induced Matching}}
\newcommand{\mif}{{\sc Maximum Induced Forest}}

\newcommand{\ois}{{\sc Optimal Induced Subgraph for $\cP$ and $t$}}
\newcommand{\cis}{{\sc Constrained Induced Subgraph for $\cP$ and $t$}}

\newcommand{\owais}{{\sc Optimal Weighted Annotated Induced Subgraph for $\cP$ and $t$}}

\newcommand{\mifd}{{\sc Maximum Induced $\mathcal{F}$-free Subgraph}}

\newcommand{\amims}{\textsc{Maximum Induced Subgraph with $\leq \ell$ copies of $\cF_m$-cycles}}

\newcommand{\mfvs}{{\sc Minimum Feedback Vertex Set}}
\newcommand{\mislkc}{\textsc{Maximum Induced Subgraph with $\leq \ell$ copies of $p$-cycles}}

\newcommand{\mislf}{{\sc Maximum Induced Subgraph with   $ \leq  \ell$ copies of Minor Models from $\mathcal{F}$}}
\newcommand{\midcls}{{\sc Minimum Induced Disjoint Connected $\ell$-Subgraphs}}

\newcommand{\kip}{{\sc $k$-in-a-Path}}
\newcommand{\kit}{{\sc $k$-in-a-Tree}}
\newcommand{\kic}{{\sc $k$-in-a-Cycle}}
\newcommand{\kig}{{\sc $k$-in-a-Graph From $\mathcal{G}(t,\varphi)$}}

\newcommand{\migp}{{\sc  Maximum Induced $\mathcal{G}(t,\cP)$-Packing}}
\newcommand{\igp}{{\sc  Independent $\mathcal{G}(t,\varphi)$-Packing}}
\newcommand{\hftg}{{\sc  Homomorphism from $t$-Treewidth Subgraph}}
\makeatletter
\def\imod#1{\allowbreak\mkern10mu({\operator@font mod}\,\,#1)}
\makeatother

\newcommand{\defproblem}[3]{
  \vspace{1mm}
\noindent\fbox{
  \begin{minipage}{0.96\textwidth}
  \begin{tabular*}{\textwidth}{@{\extracolsep{\fill}}lr} #1   \\ \end{tabular*}
  {\bf{Input:}} #2  \\
  {\bf{Task:}} #3
  \end{minipage}
  }
  \vspace{1mm}
}

  \newtheorem{theorem}{Theorem}
  \newtheorem{lemma}{Lemma}
  \newtheorem{corollary}{Corollary}
  \newtheorem{proposition}{Proposition}
  \newtheorem{definition}{Definition}

\title{Large induced subgraphs via triangulations and CMSO}

\author[1]{Fedor V. Fomin}
\author[2]{Ioan Todinca\thanks{Partially supported by the ANR project AGAPE.} }
\author[1]{Yngve Villanger}
\affil[1]{Department of Informatics, University of Bergen, Norway, \texttt{fomin@ii.uib.no, yngve.villanger@uib.no}}
\affil[2]{LIFO, Univ. Orl\'eans, France, \texttt{ioan.todinca@univ-orleans.fr}}
 \date{\today}
\begin{document}
\maketitle

\thispagestyle{empty}
\begin{abstract} 
We obtain an   algorithmic meta-theorem for the following    optimization problem. Let 
$\varphi$ be   a Counting Monadic Second Order Logic  (CMSO) formula and $t\geq 0$ be an integer. For a given graph $G=(V,E)$, the task is to maximize 
$|X|$ subject to the following: there is a set $  F\subseteq V$ such that $X\subseteq F $, 
    the subgraph $G[F]$ induced by $F$ is of treewidth at most $t$,   and structure $(G[F],X)$ models $\varphi$, i.e.  $(G[F],X)\models\varphi$. 
Special cases of this optimization problem are the following generic examples. Each of these special cases  contains various problems as a special subcase:
 \begin{itemize}
 \item \amims, where for fixed  nonnegative integers $m$ and $\ell$, the task is to find a maximum induced subgraph of a given graph with  at most $\ell$ vertex-disjoint cycles of length $0\imod{m}$. For example,  this encompasses     the problems of finding a maximum induced forest  or a maximum subgraph without even cycles.
 
\item 
\fd, where for a fixed finite set of graphs ${\cal F}$ containing a planar graph, the task is to find   a maximum induced subgraph  of a given graph containing no graph from ${\cal F}$ as a minor. Examples of \fd{} are the problems of finding a minimum  vertex cover  or a minimum number of vertices required to delete from the graph to obtain an outerplanar graph. 
\item \textsc{Independent $\mathcal{H}$-packing},  where for a fixed finite set of connected graphs ${\cal H}$, the task is 
 to find   an induced subgraph $F$ of a given graph with the maximum number of connected components, such that each connected component of $F$ is isomorphic to some graph from  $\mathcal{H}$. For example, the problem of finding  a maximum induced matching or packing into nonadjacent triangles, are the special cases of this problem. 
\end{itemize}   
   We give an algorithm solving the optimization  problem on an $n$-vertex graph $G$ in time $\cO(|\Pi_G| \cdot n^{t+4}\cdot f(t,\varphi))$, where  $\Pi_G$ is the set of all potential maximal cliques in $G$ and  $f$ is a function of $t$ and $\varphi$  only. We also show how similar running time can be obtained for the weighted version of the problem. 
   Pipelined with known bounds on the number of potential maximal cliques,  we derive   a plethora of algorithmic consequences 
   extending and subsuming many known results on algorithms for special graph classes and  exact exponential algorithms. 
\end{abstract}

\newpage
 \setcounter{page}{1}
\section{Introduction}

We provide a generic algorithmic result concerning  induced subgraphs with properties expressible in some logic. The main applications of our result can be found in two areas of graph algorithms: polynomial time algorithms on special graph classes and exponential time algorithms. 
%

%


\medskip\noindent\textbf{Graph classes.}    The algorithmic study of  graphs with particular structure can be traced to  the introduction of perfect graphs by Berge in the beginning of 1960s. Most of the research in this area 
focuses  on graph algorithms  exploiting  the structure of the input graph. Many problems intractable on general graphs were shown to be solvable in polynomial time on different classes of graphs like interval or chordal graphs.
The book of Golumbic \cite{Golumbic80} provides algorithmic studies of fundamental classes of perfect graphs while  the book of Brandst\"{a}dt  et al. \cite{brandstadt1999graph} gives an extensive overview of different classes of graphs.
By the seminal work of Gr{\"o}tschel et al. 
\cite{GrotschelLS81}, the weighted versions of  \mis, \textsc{Maximum Clique},  \textsc{Coloring}, and \textsc{Minimum Clique Cover} are solvable in polynomial time on perfect graphs. 
 There are two natural research directions in this area   extending the limits of tractability. One direction  is to identify  graph classes beyond perfect graphs,  where a specific problem like \mis, can still be solved efficiently. The second   direction is to identify more general problems which still can be solved in polynomial time on subclasses of perfect graphs. 
%

As an  example, let us take   \mif\footnote{In the literature, the complementary minimization problem of deleting the minimum number of vertices such that the remaining graphs has no cycles, is known as \mfvs. Since from exact algorithms perspective maximization and minimization versions are equivalent, we will be discussing mostly maximization problems.}, which can be seen as a natural extension of \mis, where instead of maximum edgeless graph one is seeking for a maximal acyclic graph.  It easy to notice that the problem is NP-complete being restricted to bipartite, and thus to perfect, graphs.  On the other hand,  for other classes of graphs the problem is solvable in polynomial time. 
Yannakakis and Gavril \cite{YannakakisG87} have shown how to find in polynomial time  a maximum induced forest and tree on chordal graphs.
In fact, they show polynomial time solvability of more general problem  of finding maximum and connected maximum $k$-colorable subgraphs in chordal graphs, where $k$ is a constant. When $k$ is a part of the input, they showed that on chordal graphs both problems are NP-compete. 
Other graph classes where \mif{}  was known to be solvable in polynomial time include  circle {$n$}-gon graphs, circle trapezoid, circle graphs, and bipartite chordal graphs \cite{Gavril08,Gavril11,KloksKP12}. 
The containment relations between these classes of graphs is given in Fig~\ref{fig:graph_classes}. 
  \begin{figure} 
\begin{center}
\includegraphics[scale=.42]{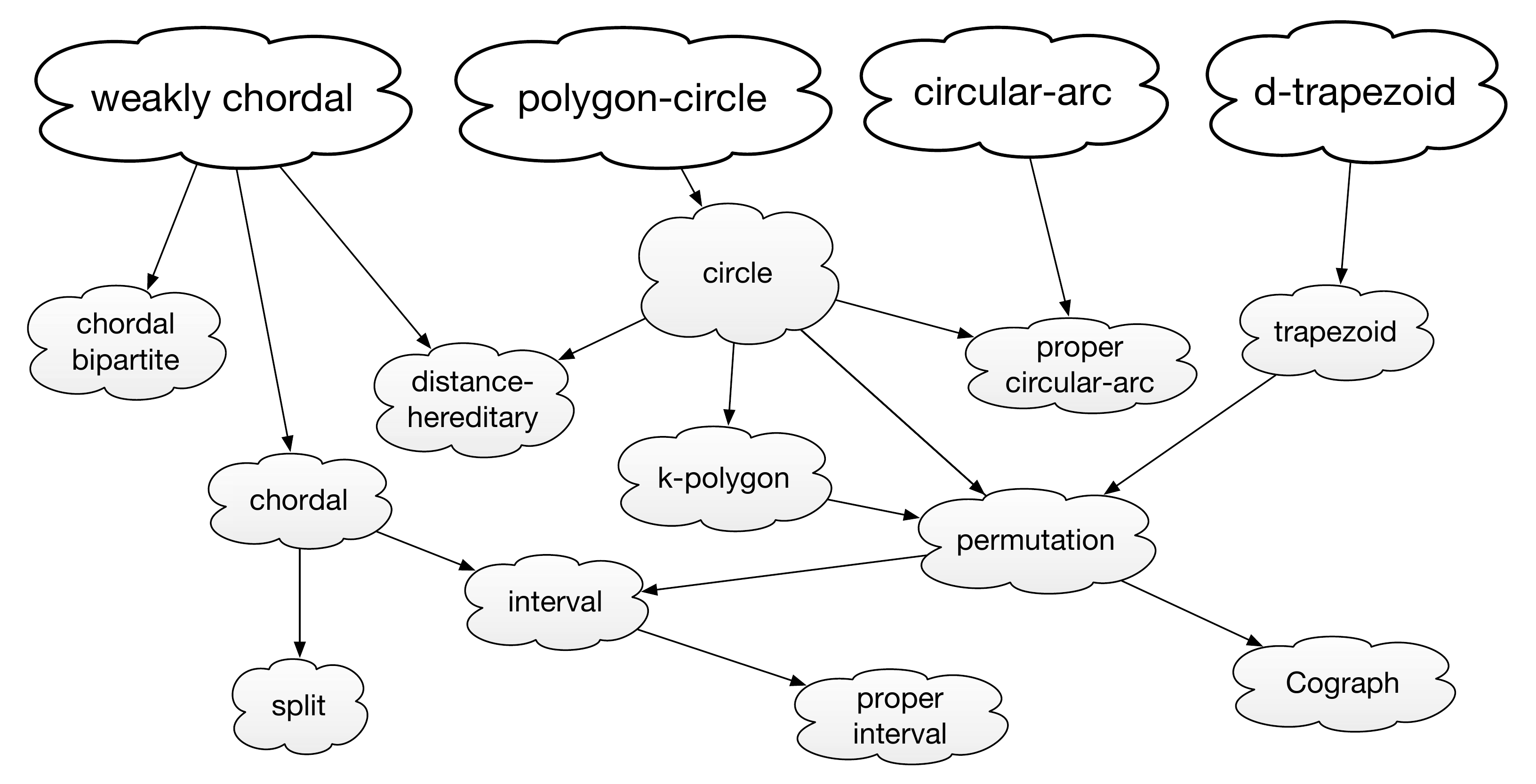}
\caption{Graph classes with a polynomial number of potential maximal cliques. }\label{fig:graph_classes}
\end{center}
\end{figure}
According to the database \texttt{http://www.graphclasses.org} on special graph classes the complexity of (weighted) \mif{} on weakly chordal is open.  

%
%
  Another example of a well-studied problem on special graph classes is \mim. Here the task is to find a maximum induced subgraph such that every connected component of this graph is an edge. 
 The complexity of this problem on different graph classes was investigated in     \cite{Cameron89,CameronST03,Chang03,GolumbicL00}.    
Cameron and Hell in \cite{CameronH06} introduced the following   generalization of \mim.
Let $\mathcal{H}$ be a finite set of connected graphs. An $\mathcal{H}$-packing of a given graph $G$ is a pairwise vertex-disjoint set of subgraphs of $G$, each isomorphic to a member of $\mathcal{H}$. An independent $\mathcal{H}$-packing of a given graph $G$ is an $\mathcal{H}$-packing,
i.e. a set of pairwise vertex-disjoint set of subgraphs of $G$, each isomorphic to a member of $\mathcal{H}$, such that 
  no two subgraphs of the packing are joined by an edge of $G$. The task is to find 
 the maximum number of graphs contained in an independent $\mathcal{H}$-packing. For example, when 
  $\mathcal{H}$ consists of $K_1$ this is \mis, and when $\mathcal{H}=\{K_2\},$ this is \mim.
  It has been shown in \cite{CameronH06}  that for many graph classes  including weakly chordal and  polygon-circle graphs, $\mathcal{H}$-packing is solvable in polynomial time.

\medskip\noindent\textbf{Exact exponential algorithms.} The second  application of our   results can be found in the area of exact exponential algorithms.
The area of exact exponential algorithms is about   solving  intractable  problems  faster than the trivial exhaustive search, though still in exponential time  \cite{FominKratschbook10}.  While for any graph property $\pi$ testable  in polynomial time, the problem of finding a maximum induced subgraph with property $\pi$ is trivially solvable in time  $2^n n^{\cO(1)}$, for several fundamental problems much faster algorithms are known. A longstanding open  question in the area is  if \textsc{Maximum Induced   Subgraph with Property $\pi$} can be solved faster than  the trivial $\cOs(2^n)$ for 
 every hereditary property $\pi$ testable  in polynomial time.

For the simplest property $\pi$, being edge-less, the corresponding maximum induced subgraph problem is \mis. 
A significant amount of  research was also devoted to algorithms for this problem starting from the classical work of Moon and Moser \cite{MoonM65} (see also  Miller and Muller \cite{MillerMuller60}) from the 1960s  
\cite{TarjanTrojanowski77,Jian86,Robson86,FominGK09-acm,BourgeoisEPR12,KneisLP09}. To the best of our knowledge, the fastest known algorithm of running time  $\cO(1.2109^n)$ is due to Robson \cite{Robson86}. For \mif{} an algorithm of running time $\cO(1.7548^n)$ was known \cite{FominGPR08-On}. This result was   improved and generalized by a  subset of the authors, who have shown that for any fixed $t$, the maximum induced subgraph of treewidth at most $t$ can be computed in time  $\cO(\npmc^n)$ \cite{Fomin:2010ys}.
   There is also a relevant work of  Gupta et al. \cite{GuptaRS12}
 who gave algorithms for 
 \textsc{Maximum Induced Matching} and \textsc{Maximum 2-Regular Induced Subgraph}, with running times  time 
 $\cO(1.695733^n)$ and $\cO(1.7069^n)$, respectively.  
  
   \medskip
 Our main theorem is based on developments from  two  research areas: the theory of minimal triangulations and logic. 
   
   \medskip\noindent\textbf{Minimal triangulations.} A triangulation of a graph $G$ is a chordal (no induced cycle of length at least four) supergraph 
   of $G$. A triangulation $H$ of $G$ is minimal, if no proper subgraph of $H$ is a triangulation of $G$.  Triangulations are closely related to 
   fundamental problems arising in sparse matrix computations which were  studied intensively in the past
  \cite{Parter61,Rose72}.   The survey of Heggernes    \cite{Heggernes06} gives an overview of techniques and applications of   minimal triangulations. It appeared in 1990s that minimal separators play important role in obtaining minimal triangulations with certain properties. Techniques based on minimal separators were used to obtain polynomial algorithm computing the treewidth and  minimum fill-in for different classes of graphs  \cite{Bodlaender:1995fk,KloksKW98,Kloks:1997uq}. These results were extended by Bouchitt\'e and Todinca in 
  \cite{BoTo01,BoTo02}, who also introduced the notion of a \pmc, which is a set of vertices of a graph that is a clique in some minimal triangulation. Potential maximal cliques appeared to be a handy tool for computing the treewidth of  a graph  \cite{FKTV08,Fomin:2012kx}.
 Recently     \pmc{} based machinery was used to obtain a subexponential parameterized algorithm finding a minimum fill-in of a graph \cite{fomin2012subexponential}. The work which is most relevant to our results is the work of a  subset of the authors   \cite{Fomin:2010ys}, where \pmc s were used to find maximum induced subgraphs of treewidth at most $t$. We build on the previous techniques exploiting  the structure of minimal triangulations, minimal separators and \pmc s but  to use   the   framework of minimal triangulations in full generality, we have to combine it with the powerful tools from logic.
 
\medskip\noindent\textbf{Algorithmic applications of logic.} Algorithmic meta-theorems are algorithmic results which can be applied  
  to large  families of combinatorial problems, instead of just specific problems.
 Such theorems provide  a better understanding of the scope of general algorithmic
techniques and the limits of tractability.
 Usually meta-theorems are based on the deep relations
between logic and combinatorial structures, which is a fundamental issue
of computational complexity \cite{Grohe07logic,Kreutzer2011}.
  A typical example of a meta-theorem is the  celebrated Courcelle's theorem
\cite{Courcelle92a} which states that all graph
properties definable in Monadic Second Order Logic   can be decided
in linear time on graphs of bounded treewidth. More recent
examples of such meta-theorems state that all first-order 
definable properties on planar graphs can be decided in linear
time~\cite{FrickG01-dec}, that all first-order definable optimization problems
on classes of graphs with excluded minors can be approximated in
polynomial time to any given approximation ratio \cite{DawarGK07},  and that all parameterized problems with finite integer index and additional ``compactness" or ``bidimensional" combinatorial property, admit linear kernels on planar graphs  \cite{H.Bodlaender:2009ng,F.V.Fomin:2010oq}.
 As it often happens with  meta-theorems, a combination of   logic and graph theory not only give a uniform 
explanation to tractability  of many algorithmic problems  but also
provide a variety of new results.
There are several extensions of Courcelle's theorem known in the literature. In particular, for a counting variant of MSO, Counting Monadic Second Order Logic (CMSO),  where we are allowed to have sentences   testing if a 
set is equal 
to $q$ modulo $r,$ for some integers $q$ and $r$, and analogue of Courcelle's theorem was obtained by Borie et al. \cite{BPT92}
and Lagergren and  Arnborg \cite{Lagergren:1991vn}. 
 Our proof is using the framework of Borie et al. \cite{BPT92}.
 


\medskip\noindent\textbf{Our results.}  
A \emph{property $\cP(G,X)$} on graphs, where $G$ is a graph and $X$ is a vertex subset of its vertices,   associates to each graph $G$ and each vertex subset $X$ of $G$ a boolean value. Borie et al. \cite{BPT92} defined \emph{regular properties},  which  definition we postpone till the next section. For all our applications, we need only the fact from Borie et al. \cite{BPT92} that 
every  property $\cP(G,X)$ expressible by a CMSO-formula is regular. Then our result can be stated as follows.  
 Let 
$\varphi$ be   a  CMSO-formula,  $G=(V,E)$ be a graph,  and $t\geq 0$ be an integer.  We consider the following optimization problem
\begin{equation}\label{eq:opt_phi} 
\begin{array}{ll}
\mbox{Max}  &  |X|   \\
\mbox{subject to } &  \mbox{There is a set }   F\subseteq V    \mbox{ such that }  X\subseteq F;       \\
 &  \mbox{The treewidth of   }  G[F]    \mbox{ is at most }  t ;     \\
 &  (G[F],X)\models\varphi.  
\end{array}
\end{equation}
 For example, \mis{} can be encoded by \eqref{eq:opt_phi} by taking $t=0$, and $\varphi$ expressing  that  $X=F$ and the  absence of edges in $G[F]$. For another example, consider \textsc{Independent Cycle Packing}, where the task is to find an induced subgraph with maximum number of connected components such that each component is a  cycle. In this case,   $t=2$ and $\varphi$ expresses the    property that each connected component is a cycle and that $X$ is a  set  of vertices containing exactly one vertex from each cycle. 
 
Let $\Pi_G$ be the set of all potential maximal cliques in  $G$.  Our main result is that  \eqref{eq:opt_phi} is solvable in time 
$O(|\Pi_G| \cdot |V|^{t+4}\cdot f(t,\varphi))$ for some function $f$. Moreover, within the same running time one can find the corresponding sets $X$ and $F$. Also it is easy to extend our algorithm to solve within the same running time  weighted and annotated versions of \eqref{eq:opt_phi}. 
  
Many well studied 
 graph classes have the following property: there is a polynomial function $p$, depending only on the graph class, such that for every graph $G$ from the class, the number of potential maximal cliques in $G$  is at most $p(n)$, see Fig~\ref{fig:graph_classes} for examples of such classes. Moreover, if the number of \pmc s in a graph is bounded by some polynomial of $n$, then all \pmc s can be enumerated in polynomial time \cite{BoTo02}. 
Our algorithm implies directly that 
 every problem expressible in the form of  \eqref{eq:opt_phi} is solvable in polynomial time on such graph classes. We discuss in details the bounds  on the number of \pmc s for different graph classes in Section~\ref{sec:graph_classes}.
  Interestingly enough,  while recognition of several of  graph classes, like polygon-circle or $d$-trapezoid, can be  NP-complete, our algorithm is still able either to solve the problem, or to report that the input graph does  not belong to the specified graph class. Such algorithms were called \emph{robust} by  Raghavan and Spinrad~\cite{RaghavanS03}. To the best of our knowledge, very few  robust algorithms were known in the literature prior to our work.

  Another direct consequence of our algorithm is that because  every $n$-vertex  graph has    $\cO(\npmc^n)$  \pmc s \cite{Fomin:2010ys},    many  intractable problems concerning maximum induced subgraphs with different properties expressible  in the form of  \eqref{eq:opt_phi}, can be solved significantly faster than by the  trivial  $\cO(2^n)$-time  brute-force algorithm. We are not aware of any algorithmic meta-result of this flavor in the area of exact algorithms.


\medskip
We mention below the most interesting  special cases of the optimization problem \eqref{eq:opt_phi}.
Each of these special cases  contains various problems as a special subcase, we discuss subcases after introducing each of the problems. 
For some of these cases, expressibility in the form of  \eqref{eq:opt_phi} is trivial but for some it is non-obvious and requires deep results from Graph Theory. We discuss these issues  in more details in Section~\ref{se:applications}.

   \medskip
  Let   $\cF_m$  be the set of cycles of length  $0\imod{m}$.
   Let  $\ell \geq 0$ be an  integer.  Our first example is the following  problem.
   
   \medskip
\defproblem{\amims{}}{A graph $G$.}{Find a  set $F\subseteq V(G)$ of maximum size such that  
   $G[F]$ contains at most $\ell$ vertex-disjoint cycles  from ${\cal F}_m$.}
 \medskip 
  
   \amims{} encompasses several interesting problems.
   For example, when $\ell=0$, the problem is to find a maximum induced subgraph without cycles divisible by $m$. For $\ell=0$ and $m=1$ this is \textsc{Maximum Induced Forest}. 
   
   For  integers $\ell \geq 0 $ and $p\geq 3$, the problem related to \amims{} is the   following.
   
   \medskip
\defproblem{\mislkc{}}{A graph $G$.}{Find a  set $F\subseteq V(G)$ of maximum size such that  
   $G[F]$ contains at most $\ell$ vertex-disjoint cycles  of length at least $p$.
   }
   
   \medskip 
   
Next example concerns properties described by forbidden minors. 
    Graph $H$ is a \emph{minor} of graph $G$   if $H$ can be obtained from a
subgraph of $G$ by a (possibly empty) sequence of edge contractions.  A \emph{model} $M$ of minor $H$ in $G$ is a minimal subgraph of $G$, where the edge set $E(M)$ is partitioned into \emph{c-edges (contraction edges)} and \emph{m-edges (minor edges)} such that the graph resulting from contracting all c-edges is isomorphic
 to $H$.
Thus, $H$ is isomorphic to a minor of $G$ if and only if there exists a model
of $H$ in $G$. For  an integer $\ell$ a finite set of graphs ${\cal F}$, we define he following generic problem.
 
\medskip


 \defproblem{\mislf{} }{A graph $G$.}{Find a set  $F \subseteq V(G)$
  of maximum size such that  $G[F]$ contains at most $\ell$ vertex disjoint  minor models of  graphs  from
    ${\cF}$.}

Even the special case with $\ell=0$, this problem and its complementary version called the \fd, encompass many different problems.
In the literature, the case $\ell=0$ was studied from parameterized and approximation perspective \cite{FominLMS12}.
%
%
%

When ${\cal F}=\{K_2\}$, a complete graph on two
vertices, this is   {\sc Maximum Independent Set}, the  problem complementary to the {\sc Minimum Vertex Cover} problem. When ${\cal F}=\{C_3\}$, a cycle on three
vertices, this is {\sc Maximum Induced Forest}.
Case $\cF=\{K_4\}$ of \mifd{} corresponds to maximum induced serial-parallel graph, $\cF=\{K_4, K_{2,3}\}$ to maximum induced  outerplanar, and case when  $\cF$ consists of a diamond graph, which is $K_4$ minus one edge, is to find a maximum induced cactus subgraph. Maximum induced 
pseudo-forest is the case of $\cF$ containing  the diamond and butterfly graphs, which is obtained by identifying one vertex of two triangles. Maximum  
 Apollonian graph corresponds to the situation with $\cF$ consisting of  the complete graph $K_5$, the complete bipartite graph $K_{3,3}$, the graph of the octahedron, and the graph of the pentagonal prism. 
A fundamental problem, which is a special case of  \fd{}, is 
 {\sc Minimum Treewidth $\eta$-Deletion} or  {\sc
$\eta$-Transversal} which is  to delete
minimum vertices to obtain a graph of treewidth at most $\eta$. Since by the result of Robertson and Seymour \cite{RobertsonS3} every graph of treewidth $\eta$ excludes a 
$(\eta+1)\times (\eta+1)$ grid as a minor, we have that the set $\cal F$ of forbidden minors of treewidth 
$\eta$ graphs contains a planar graph. 
Similarly, for $\ell>0$, \mislf{} generalizes problems like finding a maximum induced subgraph containing at most $\ell$ vertex-disjoint cycles,  at most $\ell$ vertex-disjoint outerplanar graphs, at most $\ell$ vertex-disjoint subgraphs of treewidth $t$, etc. 
 For some graph classes, like circular-arc and weakly chordal, we show that even more general cases of 
\fd, when   $\cF$ is not requested to contain a planar graph, are still solvable in polynomial time.

%
  
  \medskip

   Let $t\geq 0$ be an integer and $\varphi$ be a CMSO-formula. 
  Let $\mathcal{G}(t,\varphi)$ be a class of connected graphs of treewidth at most $t$ and with property expressible by $\varphi$.
Our next example is  the following problem.
%


%

\defproblem{\igp{}}{A graph $G$.}{Find a  set $F \subseteq V(G)$ with  maximum number of connected components such that  
   each connected component of $G[F]$ is in $\mathcal{G}(t,\varphi)$.}

\medskip

 In other words, the task is to find a maximum  vertex-disjoint packing in $G$ of  subgraphs  from   $\mathcal{G}(t,\varphi)$ such that  no two subgraphs of the packing are joined by an edge of  $G$.
 This problem trivially generalizes several well studied problems. For example,  \mim{} is to find a maximum induced matching which was studied intensively for different graph classes.  Similarly, when class $\mathcal{G}(t,\varphi)$ consists of one graph $K_3$, 
 then \migp{} is induced triangle packing.
 This problem, under the name \textsc{Independent Triangle Packing} was studied by Cameron and Hell \cite{CameronH06}. Recall that Cameron and Hell defined more general problem, namely,  \textsc{Independent $\mathcal{H}$-Packing}, 
  where for a finite set of connected graphs $\mathcal{H}$, the task is to find a maximum number of disjoint copies of graphs from $\mathcal{H}$ such that there is no edges between the copies. Since every finite set of graphs is trivially   in $\mathcal{G}(t,\cP)$ for some $t$ and $\cP$,   \textsc{Independent $\mathcal{H}$-Packing} is a special case of \igp.
  Another studied variant of the problem is \textsc{Induced Packing of Odd Cycles} introduced by Golovach et al.  in  \cite{GolovachKPT12}, where the task is to find the maximum number of odd cycles such that there is no edge between any pair of cycles. 
  
%

The next problem is an example of annotated version of optimization problem  \eqref{eq:opt_phi}.

\defproblem{\kig{}}{A graph $G$,  with $k$ terminal vertices.}{Find   an induced graph from   $\mathcal{G}(t,\varphi)$ containing all $k$ terminal vertices.}

\medskip
It is also easy to handle variants of this problem where terminal vertices have specific properties, like being the endpoints of the path if  $\mathcal{G}(t,\varphi)$ is the class of paths. 
Many variants of \kig{} can be found in the literature, like  \kip,\ \kit, \kic.
 \kip{} is clearly solvable in polynomial time for $k=2$. For $k=3$ the problem is NP-complete already on graph of maximum vertex degree at most three \cite{DerhyP09}. Bienstock \cite{Bienstock91} have shown that the following cases of \kig{} are NP-hard:
finding an induced odd cycle of length greater than three, passing through a prescribed vertex and 
finding an induced odd path between two prescribed vertices.  Polynomial time algorithms for   the odd path problem are known for several graph classes, including chordal  \cite{ArikatiP93} and  circular-arc graphs \cite{ArikatiRM91}. 
Chudnovsky and Seymour have shown that  \kit{}   for $k=3$ is solvable in polynomial time \cite{ChudnovskyS10}. The complexity of the case $k=4$ is open. 

\medskip 

Let us remark that because of the power of CMSO, different modifications of the problems mentioned above, with additional requirements   on the  induced subgraph like being connected,   constrains on vertex degree and parities of connected components,  can be easily handled. 
\section{Preliminaries}
We denote by $G=(V,E)$ a finite, undirected and simple graph
with $|V|=n$ vertices and $|E|=m$ edges.
Sometimes  the vertex set of a graph $G$ is referred to as $V(G)$ and its edge set as $E(G)$.  
A clique $K$ in $G$ is a set of pairwise adjacent vertices of $V(G)$. 
The \emph{neighborhood} of a vertex $v$ is
$N(v)=\{u\in V:~\{u,v\}\in E\}$.
For a vertex set $S\subseteq V$ we denote by $N(S)$ the set $\bigcup_{v \in S} N(v)\sm S$.

The notion of
treewidth is due to Robertson and Seymour \cite{RobertsonS3}. A {\em
tree decomposition} of a graph $G=(V,E)$, denoted by $TD(G)$, is a
pair $({X}, T)$, where   $T$ is a tree and ${X}=\{{X}_i
\mid i\in V(T) \}$ is a family of subsets of $V$, called  \emph{bags},   such that
\begin{itemize}
\item[(i)] $\bigcup_{i\in V(T)}{X}_i= V$,
\item[(ii)]  for each edge $e=\{u,v\} \in E(G)$ there  exists   $i\in V(T)$ such that both $u$ and $v$
are in  ${X}_i$, and
\item[(iii)]  for all $v\in V$, the set of nodes
$\{i\in V(T) \mid v \in {X}_i\}$ induces a connected subtree of $T$.
\end{itemize}
The maximum  of $|{X}_i|-1$, $i\in V(T)$,  is called the {\em width} of the
tree decomposition. The {\em treewidth} of a graph $G$, denoted by $\tw(G)$,
is the minimum width taken over all tree decompositions of $G$.

\medskip
\noindent
\textbf{Counting Monadic Second Order Logic.} 
We use Counting Monadic Second Order Logic (CMSO), an extension of MSO, as a basic tool to express properties of vertex/edge sets in graphs. 
 \smallskip

The syntax of Monadic Second Order Logic (MSO) of graphs includes the logical connectives $\vee,$ $\land,$ $\neg,$ 
$\Leftrightarrow ,$  $\Rightarrow,$ variables for 
vertices, edges, sets of vertices, and sets of edges, the quantifiers $\forall,$ $\exists$ that can be applied 
to these variables, and the following five binary relations: 
\begin{enumerate}

\item 
$u\in U$ where $u$ is a vertex variable 
and $U$ is a vertex set variable; 
\item 
 $d \in D$ where $d$ is an edge variable and $D$ is an edge 
set variable;
\item 
 $\mathbf{inc}(d,u),$ where $d$ is an edge variable,  $u$ is a vertex variable, and the interpretation 
is that the edge $d$ is incident with the vertex $u$; 
\item 
 $\mathbf{adj}(u,v),$ where  $u$ and $v$ are 
vertex variables  and the interpretation is that $u$ and $v$ are adjacent; \item 
 equality of variables representing vertices, edges, sets of vertices, and sets of edges.
\end{enumerate}

 In addition to the usual features of monadic second-order logic, if we have atomic sentences testing whether the cardinality of a set is equal 
to $q$ modulo $r,$ where $q$ and $r$ are integers such that $ 0\leq q<r $ and $r\geq 2,$ then 
this extension of the MSO is called the {\em counting monadic second-order logic}. So essentially CMSO 
is MSO with the following atomic sentence for a set $S$: 
\begin{quote}
$\mathbf{card}_{q,r}(S) = \mathbf{true}$ if and only if $|S| \equiv q \pmod r.$ 
\end{quote}
We refer to~\cite{ArnborgLS91,Courcelle90,Courcelle97} and the book of Courcelle and Engelfriet~\cite{CoEn12} for a detailed introduction on CMSO. In~\cite{CoEn12}, the CMSO is referred to as $CMS_2$.
%

\subsection{Treewidth, $t$-terminal recursive graphs and regular properties}

We use one of the (many) alternative definitions of treewidth, based on \emph{terminal graphs}.
A \emph{$t$-terminal graph $G=(V,T,E)$} is a graph with an ordered set $T \subseteq V$ of at most $t$ distinguished vertices, called \emph{terminals}. Denote by $\tau(G)$ the number of terminals of graph $G$.

A $t$-terminal graph $(V,T,E)$ is a \emph{base graph}  if $V = T$. We define \emph{composition operations} over the set of $t$-terminal graphs. A composition operation $f$ is of arity $1$ or $2$. When $f$ is of arity $2$, it acts on two $t$-terminal graphs $G_1, G_2$ and produces a $t$-terminal graph $G=f(G_1, G_2)$ as follows.  It first makes disjoint copies of the two graphs, then ``glues'' some terminals of graphs $G_1$ and $G_2$. 
Operation $f$ is represented by a matrix $m(f)$. The matrix has 2 columns and $\tau(G) \leq t$ lines, its values are integers between $0$ and $t$. At line $i$ of the matrix, elements $m_{ij}(f)$ indicate which terminals of graphs $G_j$ are identified to terminal number $i$ of $G$. If $m_{ij}(f) = 0$ it means that no terminal of $G_j$ was identified to terminal number $i$ of $G$. A terminal of $G_j$ can be identified to at most one terminal of $G$ (a column $j$ cannot contain two equal, non-zero values). Note that if $m_{i1}(f)=0$ and $m_{i2}(f)=0$ it means that terminal  $i$ of $G$ is a new vertex.

When $f$ is of arity 1, its matrix $m(f)$ has only one column. The $t$-terminal graph $G = f(G_1)$ is obtained from graph $G_1$ and matrix $m(f)$ as above, by identifying terminal $m_{i1}(f)$ to terminal number $i$ in $G$. 

Observe that the number of possible composition operations over $t$-terminal graphs is bounded by some function of $t$.
We say that a $t$-terminal graph $G$ is \emph{$t$-terminal recursive} if it can be obtained from $t$-terminal base graphs through a sequence of composition operations. This sequence is called the \emph{$t$-expression} of graph $G$. 

\begin{proposition}[\cite{Bodlaender98}]\label{pr:tw_expr}
For any $(t+1)$-terminal recursive graph $H=(V,T,E)$, there is a tree decomposition of $(V,E)$  of width at most $t$, with a bag containing $T$.
Conversely, for any tree decomposition of width $t$ of graph $G=(V,E)$ and any bag $W$ of the decomposition, $(V,W,E)$ is a $(t+1)$-terminal recursive graph.
\end{proposition}
\begin{proof}
Assume that $(V,T,E)$ can be obtained recursively, through composition operations, from $(t+1)$-terminal base graphs. The expression constructing this graph can be represented as a tree, the leaves being the base graphs, each internal node corresponding to a composition operation. The tree decomposition of $G$ is simply obtained by following this tree and putting, in each node, a bag corresponding to the terminals of the graph represented by the corresponding sub-expression. The bags are clearly of size at most $t+1$.  One can easily check that the set of bags satisfies the conditions of a tree decomposition. 

The other direction is proved in~\cite{Bodlaender98}, Theorem~40.
\end{proof}


Consider a \emph{property $\cP(G,X)$} on graphs  depending on a vertex subset $X$. That is,  property $\cP$ associates to each graph $G$ and each vertex subset $X$ of $G$ a boolean value. 
By the celebrated results of~\cite{Courcelle90,ArnborgLS91,BPT92}, it is well-known that if the property can be expressed by a CMSO-formula, there exists a linear-time algorithm taking as input a $(t+1)$-terminal recursive graph $G=(V,T,E)$ and computing a maximum (or minimum) size vertex set $X$ such that $\cP(G,X)$. Many natural problems like \textsc{Maximum Independent Set} or \textsc{Minimum Dominating Set} can be expressed in this setting. 

Typical algorithms for such problems proceed by dynamic programming. When browsing the $(t+1)$-expression of $G$, the algorithm stores in each node a table of \emph{classes} (sometimes called \emph{characteristics}) depending on the branch of the current sub-expression and the partial solutions (i.e., possible subsets of $X$) encountered so far. Let $G_1$ be such a sub-expression and let $X_1$ be a subset of vertices that we aim to extend into the solution $X$. The intuition is that if the class of $(G_1,X_1)$ is the same as the class of some other pair $(G_2,X_2)$, then we can replace the branch of $G_1$ by an expression of $G_2$, and the new graph $G'$ is such that $X_1$ extends into a solution $X_1 \cup Y$ of $G$ if and only if $X_2$ extends into a solution $X_2 \cup Y$  of $G'$. 

In order to efficiently solve our problem, we need an efficient computation of classes for base graphs, as well as an efficient computation of the classes for compositions of graphs and partial solutions. 

%

We give a formal definition of these ``good'' properties; the vocabulary is inspired by Borie \textit{et al.}~\cite{BPT92}.


Let now $G=(V,T,E)$ be a $(t+1)$-terminal recursive graph. 
For any composition operation $f$, let $\circ_f$ denote the composition operation over pairs $(G,X)$, where $f$ extends in a natural way over the values of vertex sets. If $G = f(G_1)$ then $ \circ_f((G_1,X)) = (G,X)$. 
If $G = f(G_1,G_2)$ then $ \circ_f((G_1,X_1),(G_2,X_2)) = (G,X)$, the operation being valid only if, for each terminal of $G$, either we have mapped terminals from both $G_1$ and $G_2$, contained in both $X_1$ and $X_2$, or we have not mapped any terminal belonging to $X_1$ or $X_2$. Then $X$ is obtained from  $X_1$ and $X_2$ by merging those vertices corresponding to terminals that have been mapped on a same terminal of $G$.

\begin{definition}[Regular Property]\label{de:regular}
Consider a property $\cP(G,X)$ over graphs and corresponding vertex subsets. Property $\cP$ is called \emph{regular} if, for every $t$, there exists a finite set $\cC$, a homomorphism $h$ associating to each $(t+1)$-terminal recursive graph $G$ and every  $X \subseteq V(G)$ a class $h(G,X) \in \cC$,  and an \emph{update function} $\odot_f: \cC \times \cC \rightarrow \cC$ for each composition operation $f$ of arity $2$ (resp. $\odot_f: \cC\rightarrow \cC$ for each composition operation $f$ of arity $1$), satisfying:
\begin{itemize}
\item (property $\cP$ is preserved) If $h(G_1,X_1) = h(G_2,X_2)$ then  $\cP(G_1,X_1)=\cP(G_2,X_2)$.
\item (integrity of operations) For any composition operation $f$, we have that $$h(\circ_f((G_1,X_1),(G_2,X_2))) = \odot_f(h(G_1,X_1),h(G_2,X_2))$$ if $f$ is of arity $2$, and $$h(\circ_f(G_1,X_1)) = \odot_f(h(G_1,X_1))$$ if $f$ is of arity $1$.
\end{itemize}
\end{definition}
We point out that the homomorphism class $h(G,X)$ depends on $G$ and on the value of $X$. 
Typically the class of $h(G,X)$ encodes, among other informations, the intersection of $X$ with the set of terminals. 
For example, if the composition operation $\circ_f((G_1,X_1),(G_2,X_2))$ is not valid, then
$\odot_f(c_1,c_2)$, where $c_1$ and $c_2$ are the respective homomorphism classes of $(G_1,X_1)$ and of $(G_2,X_2)$, is also undefined.

Note that for any fixed $t$ and any regular property $\cP$, the number of classes is constant. Nevertheless, this constant depends on $t$ and on the property $\cP$. For algorithmic purposes, given $t$ and $\cP$, we need an explicit algorithm computing the homomorphism class of a given base graph, and an algorithm computing the update functions $\odot_f$. I.e., we need an algorithm that takes as input a composition operation $f$ and one or two classes $c_1, c_2 \in \cC$ and computes the class $\odot_f(c_1, c_2)$ if $f$ is of arity 2 (resp. $\odot_f(c_1)$ if $f$ is of arity 1). Eventually, we must know the set of \emph{accepting classes}, that is the set of classes $c$ such that $h(G,X)=c$ implies that $\cP(G,X)$.

As an example, consider the property $3COL(G,X)$ which is true only if $G[X]$ is 3-colourable. We show that it is regular. 
Let $P_3(t)$ be the set of partitions of subsets of $\{1, 2 \dots, t+1\}$ into three parts. The set $\cC$ of homomorphism classes is  $P_3(t)$. Consider a $(t+1)$-terminal recursive graph $G=(V,T,E)$ and let $X \subseteq V$. 
For each 3-partition $(X_1, X_2, X_3)$ of the vertex subset $X$ into three independent sets, let $p(X_1,X_2,X_3) \in P_3(t)$ be the $3$-partition of $T \cap X$ corresponding to $(T \cap X_1, T \cap X_2, T \cap X_3)$; here, for $T \cap X_i$, we only keep the ranks of the terminals  of $T \cap X_i$ in the ordered set $T$. The class $h(G,X)$ will be $\{ p(X_1, X_2, X_3) \mid (X_1, X_2, X_3) \text{~is a partition of~} V \text{~into three independent sets}\}$. In particular, the unique non-accepting class is $\emptyset$. It is not hard to see that, for fixed $t$, the class of every base graph can be computed in constant time, and that for any composition operation $f$ the update function $\odot_f$ exists and can also be computed in constant time. The number of classes is constant even though the number of subsets $X$ is arbitrarily large. 
When solving the problem $\max |X| : 3COL(G,X)$ on a $(t+1)$-terminal recursive graph $G$, we must store, in each node $u$ of the $(t+1)$-expression, for each class $c$, the size of the maximum vertex subset $X_u$ of the current graph $G_u$ such that $h(G_u,X_u)=c$. The overall solution is the maximum one among the accepting classes of the root node.

We say that a CMSO-formula $\varphi$ \emph{expresses} a property $\cP(G,X)$ if $\cP(G,X)$ is true if and only if $(G,X)$ models $\varphi$ (i.e., the formula is true exactly on graphs $G$ and vertex subsets $X$ such that $\cP(G,X)$ is true).

\begin{proposition}[Borie \textit{et al.}~\cite{BPT92}]\label{prop:borie}
Any property $\cP(G,X)$ expressible by a CMSO-formula is regular.
\end{proposition}

Moreover, the result of Borie \textit{et al.}~\cite{BPT92} is constructive in the sense that, given a CMSO-formula, it provides the homomorphism classes $\cC$, the subset of accepting classes and the algorithms computing the classes of base graphs as well as the update functions for the regular  property $\cP$ on $(t+1)$-terminal recursive graphs. 
The regularity is actually proven in~\cite{BPT92}  for all properties expressible by CMSO-formulae, which allows an arbitrary number of free variables over vertices, edges, vertex sets and and edge sets. For our applications, it is sufficient to consider properties over graphs and one vertex set, corresponding to formulae with a unique free variable, which is a set of vertices. 

To our knowledge, the question whether all regular properties are CMSO-expressible is still open.

\subsection{Treewidth, minimal triangulations and potential maximal cliques}
\paragraph{Chordal graphs and clique trees}

A graph $H$ is {\em chordal} (or {\em triangulated}) if every
cycle of length at least four has a chord, i.e., an edge between
two nonconsecutive vertices of the cycle. 
By a classical result due to Buneman and Gavril  \cite{Buneman74,Gavril74}, 
every chordal graph $G$  has a tree decomposition   
such that each bag of the decomposition is a maximal clique of $G$.
Such a tree decomposition is referred as a \emph{clique tree} of the chordal graph $G$. 

\paragraph{Minimal triangulations, potential maximal cliques and minimal separators}
A {\em triangulation} of
a graph $G=(V,E)$ is a chordal graph $H = (V, E')$ such that $E
\subseteq E'$. Graph $H$ is a {\em minimal triangulation} of $G$ if
 for every
edge set $E''$ with $E \subseteq E'' \subset E'$, the
graph $F=(V, E'')$ is not chordal.
It is well known that
for any graph  $G$,  $\tw(G)\le k$ if and only if there is a triangulation
$H$ of $G$ of clique size at most $k+1$.

Let $u$ and $v$ be two non adjacent vertices of a graph $G$. A set of
vertices $S \subseteq V$
is a {\em $u,v$-separator} if $u$ and $v$  are in different connected components
of the graph $G[V(G) \sm S]$. A connected component $C$ of $G[V(G) \sm S]$
is a {\em full component associated to $S$} if $N(C)=S$. Separator 
$S$ is a {\em minimal
$u,v$-separator} of $G$ if no proper subset of $S$ is a $u,v$-separator.
Notice that a minimal separator can be
strictly included in another one, if they are minimal separators for different pairs of vertices. 
If $G$ is chordal, then for any minimal separator $S$ and any clique tree $T_G$ of $G$
there is an edge $e$ of $T_G$ such that $S$ is the intersection of the maximal cliques 
corresponding to endpoints of $e$~\cite{Buneman74,Gavril74}. We say that 
$S$ \emph{corresponds} to $e$ in $T_G$.

We will need the following result of  Berry et al.~\cite{BerryBC00}.
\begin{proposition}[\cite{BerryBC00}]\label{th:listing_minsep}
There is an algorithm listing the set $\Delta_G$ of all  minimal separators of an input
graph $G$ in time  $\cO(n^3|\Delta_G|)$.
\end{proposition}

A set of vertices $\Omega \subseteq V(G)$ of a graph $G$ is called a
{\em potential maximal clique} if there is a minimal triangulation
$H$ of $G$ such that $\Omega$ is a maximal clique of $H$. 

\begin{proposition}[\cite{BoTo02}]\label{pr:listing_pmc}
Let $\Pi_G$ denote the set of all potential maximal cliques of graph $G$. We have $|\Pi_G| \leq n |\Delta_G|^2 + n|\Delta_G|+1$, and the set $\Pi_G$ can be listed in time
$\cO(n^2m|\Delta_G|^2)$.
\end{proposition}

We also have:
\begin{proposition}[\cite{Fomin:2010ys}]\label{pr:listing_pmc_gen}
The set of potential maximal cliques can be listed in time
$O(\npmc^n)$.
\end{proposition}

Let $\Omega$ be a \pmc. By~\cite{BoTo01}, a subset $S \subseteq \Omega$ is a minimal separator of $G$ if and only if $S$ is the
neighborhood of a connected component of $G [V(G) \setminus \Omega]$. 

For a minimal separator $S$ and a full connected component  $C$ 
 of $G[V(G) \sm S]$,
 we say that $(S,C)$ is a {\em full block} associated to
$S$. We sometimes use the notation $(S,C)$ to denote the set of
vertices $S \cup C$ of the block. It is easy to see that if $X
\subseteq V$ corresponds to the set of vertices of a block, then
this block $(S,C)$ is unique: indeed, $S = N(V \sm X)$ and $C=X
\sm S$. 
For convenience, the couple $(\emptyset, V)$ is also considered as a full block.
For a minimal separator $S$, a full block  $(S,C)$,  and a potential maximal clique 
$\Omega$, we call the triple $(S,C, \Omega)$ \emph{good} if 
$S\subseteq \Omega \subseteq C\cup S$. By~\cite{FKTV08}, the number of good triples is at most $n|\Pi_G|$.
%

The following proposition was obtained by Fomin and Villanger~ \cite{Fomin:2010ys}.

\begin{proposition}[\cite{Fomin:2010ys}]\label{pr:CompTr}
Let $G[F]$ be an induced subgraph of a graph $G$, 
let $TF$ be a minimal triangulation of $G[F]$. There exists a minimal triangulation $TG$  of $G$ such that $TF$ is an induced subgraph of $TG$.

Equivalently, for every clique $K_G$ of $TG$, the set $K_G\cap F$ is a (possibly empty) clique of $TF$.
\end{proposition}

Moreover, they consider the problem of finding a maximum induced subgraph of treewidth at most $t$:


\begin{proposition}[\cite{Fomin:2010ys}]\label{pr:IndSubgrTwT}
Given a graph $G$ and with its set $\Pi_G$ of potential maximal cliques, problem \mitwt{} can be solved in time $O(|\Pi_G| n^{t+4})$.
\end{proposition}

By Propositions~\ref{pr:IndSubgrTwT}, \ref{pr:listing_pmc} and \ref{pr:listing_pmc_gen}, we deduce that for fixed $t$ the problem can be solved in $O(\npmc^n)$ time for arbitrary graphs, and in polynomial time for classes of graphs with polynomial number of minimal separators.


%
%

\section{Optimal induced subgraph for $\cP$ and $t$}\label{sec:MainSection}

Let $t\geq 0$ be an integer and $\cP(G,X)$ be a  property. We define the following generic problem. 

\defproblem{\ois{}}{A graph $G$}{Find sets $X \subseteq F\subseteq V$ such that $X$ is of maximum size,  the  induced subgraph $G[F]$ is of treewidth at most $t$ and $\cP(G[F],X)$ is true.}
  
  \medskip
  
Let us give two examples of  problems that are particular cases of \ois{}, when $\cP(G,X)$ is  a regular  property. 
\begin{enumerate}
\item Let ${\mathcal F}$ be a finite family of graphs containing at least one planar graph. The problem  \textsc{Maximum induced ${\mathcal F}$-minor free graph} takes as input a graph $G$ and asks for an induced subgraph $G[F]$ such that $G[F]$ contains no minor from ${\mathcal F}$, and $F$ is of maximum size for this property. As we shall see in details in Section~\ref{se:applications}, the property $\cP(G[F],X)$ expressing the fact that $G[F]$ is  ${\mathcal F}$-minor free and $X=F$ is the vertex set of $G[F]$ can be expressed by a CMSO. Since ${\mathcal F}$ contains a planar graph, $G[F]$ must be of treewidth at most $t$ for some constant $t$ depending only on ${\mathcal F}$~\cite{RobertsonS-V}. Therefore, this problem (or the equivalent problem \fd{}) is a particular case of \ois.
\item The problem \textsc{Independent $\mathcal{H}$-Packing} was introduced by Cameron and Hell~\cite{CameronH06}. Here
 $\mathcal{H}$ denotes a finite set of connected graphs, and the task is to find, in an input graph $G$, a maximum number of disjoint copies of graphs from $\mathcal{H}$ such that there is no edges between the copies. Clearly these copies induce a subgraph $G[F]$ of bounded treewidth. We will give a CMSO-formula expressing the property $\cP(G[F],X)$, which is true if and only if $G[F]$ is a collection of copies of $\mathcal{H}$, and $X$ has exactly one vertex in each connected component of $G[F]$. This problem, generalizing the \textsc{Maximum Induced Matching}, is again a particular case of \ois.
\end{enumerate}

We prove here the main theorem of this article. 

\begin{theorem}\label{theorem_main}
For any fixed $t$ and any regular property $\cP$, the problem \ois{} is solvable in $|\Pi_G| n^{t+\cO(1)}$ time, when $\Pi_G$ is given in the input. 
\end{theorem}

Let us note that by Proposition~\ref{prop:borie}, results of Theorem~\ref{theorem_main} hold for   every property $\cP(G,X)$ expressible by a CMSO-formula.
Combined with Propositions~\ref{pr:listing_pmc} and~\ref{pr:listing_pmc_gen}, we obtain the following application of 
Theorem~\ref{theorem_main}.  
\begin{corollary}
For any fixed $t$ and regular property $\cP$, problem \ois{} can be solved in $\cO(\npmc^n)$ time for arbitrary graphs, and in polynomial time for classes of graphs with polynomial number of minimal separators. 
\end{corollary}

\subsection{Notations and data structures}
Our algorithm proceeds by dynamic programming on blocks and on good triples. The general 
strategy of dynamic programming over  blocks and good triples follows the ideas from 
\cite{FKTV08} and \cite{Fomin:2010ys} for computing the treewidth and subgraphs of bounded treewidth. 
However, the devil is in details, and we need more work to make this strategy applicable for our problem.

Recall that in our definition of $(t+1)$-terminal graphs, the set of terminals is ordered. The vertices of our graph are numbered from $1$ to $n$. An ordered set $W$ of vertices corresponds to
this natural ordering over set $W$. Property $\cP$ is regular, so notations $\cC$, $h$ and $\odot_f$ correspond to Definition~\ref{de:regular}.

Let $G[F]$ be an induced subgraph of $G$ and let $TF$ be a triangulation
of $G[F]$. We say that a minimal triangulation $TG$ of $G$ \emph{respects} $TF$ if, for any clique $K$ of $TG$, its intersection with $F$ is a clique in $TF$. By Proposition~\ref{pr:CompTr}, if $G[F]$ is of treewidth at most $t$, then there exists a (minimal) triangulation $TF$ of $G[F]$ of width at most $t$, and a minimal triangulation $TG$ of $G$ respecting $TF$.

The next definition and the following notations are crucial for our algorithm. 
\begin{definition}[Partial Compatible Solution]\label{de:partsol}
Let $(S,C)$ denote a full block and $(S, C, \Omega)$ denote a good triple. Let $W \subseteq S$ (resp. $W \subseteq \Omega)$ be a vertex subset of size at most $t+1$ and $c \in \cC$ be a homomorphism class for property $\cP$. We say that $(G[F],X)$ is a \emph{partial solution compatible with} $(S,C, W, c)$ (resp. with  $(S, C, \Omega, W, c)$) if:
\begin{enumerate}

\item $F \subseteq S \cup C$ and $F \cap S = W$ (resp. $F \cap \Omega = W$);
\item the $(t+1)$-terminal recursive graph $H=(F, W, E(G[F]))$ satisfies $h(H,X)=c$;
\item there is a  triangulation $TF$ of $G[F]$ of width at most $t$ and a minimal triangulation $TG$ of $G$ respecting $TF$, such that $S$ is a minimal separator (resp. $\Omega$ is a maximal clique) of $TG$.
\end{enumerate}
\end{definition}
The third condition implies that $W$ is a clique in the triangulation $TF$ of $G[F]$. 

Let $\alpha(S,C,W,c)$ (resp. $\beta(S, C, \Omega, W, c)$) denote the size of a largest vertex subset $X$ such that $(G[F],X)$ is a partial solution compatible with $(S,C, W, c)$ (resp. compatible with $(S, C, \Omega, W, c)$).
Observe that in the $\beta$ function,  $W$ represents the intersection between the partial solution and the potential maximal clique $\Omega$, while in the definition of the $\alpha$ function, $W$ is the intersection of the partial solution with the minimal separator $S$. If partial compatible solutions do not exist, we simply set $\alpha$ or $\beta$ to $-\infty$.

\subsection{The algorithm}


Our algorithm proceeds by dynamic programming on full blocks and good triples. By~\cite{FKTV08}, the number of good triples is $O(n|\Pi_G|)$.
The blocks are first sorted by size. For each block $(S,C)$ by increasing size, we first compute the values $\beta(S,C,\Omega,W,c)$ from values $\alpha(S_i,C_i,W_i,c_i)$ corresponding to smaller blocks, then we compute the values $\alpha(S,C,W,c)$ from values $\beta(S,C,\Omega,W',c')$, as described in Algorithm~\ref{al:mainalg}.

		\begin{algorithm}[H]
		\caption{  \ois{}}\label{al:mainalg}
		\SetKwInput{Input}{Input}
		\SetKwInput{Output}{Output}
		\SetKw{all}{all}
		\Indm
		\Input{graph $G$ and $\Pi_G$} 
		\Output{ sets $X \subseteq F \subseteq V(G)$ such that  $G[F]$ has treewidth at most $t$, $\cP(G[F],X)$ is true and $X$ is of maximum size}
		\Indp
		\smallskip
		Order all full blocks by inclusion\;
		\ForAll{full blocks $(S,C)$ in this order}
		{
			\ForAll{good triples $(S,C,\Omega)$, \all $W \subseteq \Omega$ of size $\leq t+1$ and \all $c \in \cC$}
			{
				\eIf{$\Omega = S \cup C$}
				{
					Compute $\beta(S,C,\Omega,W,c)$ using Equation~\ref{eq:base_case}\;
				}
				{
					Compute $\beta(S,C,\Omega,W,c)$ using Equations~\ref{eq:delta}, \ref{eq:gamma1}, \ref{eq:gammai},  and~\ref{eq:beta} \;
				}
			}
			\ForAll{$W \subseteq S$ of size $\leq t+1$ and \all $c \in \cC$}
			{
				Compute $\alpha(S,C,W,c)$ using Equation~\ref{eq:alpha}\;
			}

		}
		Compute the optimal solution using Equation~\ref{eq:final}\;
		\end{algorithm}

Consider a $(t+1)$-terminal recursive graph $D=(V_D,T,E_D)$ and let $c$ be a homomorphism class. Although this is not explicitly required by the definition of regular properties (Definition~\ref{de:regular}), we may assume w.l.o.g. that all sets $Y$ such that $h(D,Y)=c$ have the same intersection with the set $T$ of terminals. (Otherwise, if sets $Y$ and $Y'$ have different intersections with $T$ but $h(D,Y)=h(D,Y')=c$, we can ``split" class $c$ in at most $2^{t+1}$ classes, one for each possible intersection between $T$ and such a vertex subset $Y$.) 
Moreover the class $c$ encodes the intersection of $Y$ with the set of terminals of $D$, i.e., given the homomorphism class $c$, we can retrieve the rank of the vertices of $Y \cap T$. 

Therefore we assume that we have a function $term(c,T)$, taking a class $c$ and an ordered set $T$ of terminals, and returning the terminals that belong to $Y$, for any $Y$ such that $h(D,Y)=c$.  

\paragraph{The base case.} The base case consists in minimal full blocks $(S, C, \Omega)$, in which case $\Omega = S \cup C$ by~\cite{BoTo01}. In this situation, for any partial solution $(G[F],X)$ compatible with $(S,C,\Omega,W,c)$ we must have $F=W$, hence $G[W]$ corresponds to a base $(t+1)$-terminal graph. Also, we must have $X = term(c,W)$, so $X$ is unique (or might not exist).

\begin{equation}\label{eq:base_case}
\beta(S, C, \Omega, W, c) =  \left\{ \begin{array}{l l}
		 |X| &\text{~if there is~} X \subseteq W \text{~such that~}h(G[W],X) = c \\
		-\infty &\text{~otherwise} 
		\end{array}
		\right.
\end{equation}

The computation of each value $\beta(S, C, \Omega, W, c)$ corresponding to a base case takes $O(n)$ time, because we have to store the value in a table indexed by $(S, C, \Omega,c)$. 
The number of good triples is $O(n|\Pi_G|)$ so altogether these computations take $O(n^{t+3}|\Pi_G|)$ time. (Actually, one can prove by a more careful analysis that the number of good triples corresponding to base cases is at most $n$.)


\paragraph{Computing $\alpha$ from $\beta$.}
Our goal is to compute $\alpha(S,C,W,c)$ from values $\beta(S,C,\Omega,W',c')$ such that $(S,\Omega,C)$ is a good triple and $W = W' \cap S$. 

Consider any partial solution $(G[F],X)$ compatible with $(S,C,W,c)$. Let $TF$ be a  triangulation of $G[F]$ like in Definition~\ref{de:partsol} and let $TG$ be a  minimal triangulation of $G$ respecting $TF$. Let $\Omega$ be the  maximal clique of $TG$ such that $S \subseteq \Omega \subseteq S \cup C$ (this clique is unique by~\cite{BoTo01}) and take $W' = \Omega \cap F$.  Note that $(G[F],X)$ is also a partial compatible solution for $(S,C,\Omega, W', c')$ where $c'$ is the homomorphism class of $h(H',X)$; here $H'$ is the $(t+1)$-terminal recursive graph $(F, W', E(G[F]))$. Also observe that  the $(t+1)$-terminal graph $H=(F,W,G[F])$ is obtained from $H'$ by the unary composition operation $f(W',W)$ that consists in removing $W' \sm W$ from the set of terminals, and possibly renumbering the remaining terminals. Therefore  $\odot_{f(W',W)}(c')=c$.

We claim that:
\begin{lemma}\label{le:a_from_b}
\begin{equation}\label{eq:alpha}
	\alpha(S,C,W,c) = \max \beta(S,C,\Omega,W',c'),
\end{equation}
where the maximum is taken over \pmc s $\Omega$ such that $(S,C,\Omega)$ is a good triple, all subsets $W' \subseteq \Omega$ of size at most $t+1$ such that $W' \cap S = W$ and all classes $c' \in \cC$ such that 
$\odot_{f(W',W)}(c')=c$.
\end{lemma}
\begin{proof}
By the above observation, $\alpha(S,C,W,c)$ is at most the right-hand side of the equality. Conversely, let $(S,C,\Omega,W',c')$ be the quintuple realizing the maximum value of the right-hand side expression. Let $(G[F],X)$ be a partial solution compatible with $(S,C,\Omega,W',c')$. Observe that $(G[F],X)$ is also a partial solution compatible with $(S,C,W,c)$, hence $\alpha(S,C,W,c) \geq |X|$. This proves the correctness of the formula computing  $\alpha(S,C,W,c)$.
\end{proof}

For computing all values $\alpha(S,C,W,c)$ from values $\beta(S,C,\Omega,W',c')$, we proceed in a slightly different and more efficient way than the one described in the Algorithm~\ref{al:mainalg}. When 
$\beta(S,C,\Omega,W',c')$ is computed (lines 5 or 7 of the algorithm), if $\odot_{f(W',W)}(c')=c$ we simply update the value of $\alpha(S,C,W,c)$ by taking the maximum between the previous value and $\beta(S,C,\Omega,W',c')$. This only costs an extra $O(n)$ for each quintuple $(S,C,\Omega,W',c')$. The number of such quintuples is $\cO(n^{t+2}|\Pi_G|)$, thus the total cost of these computations is  $\cO(n^{t+3}|\Pi_G|)$.

\paragraph{Computing $\beta$ from $\alpha$.}
We now compute $\beta(S,C,\Omega,W,c)$ from values $\alpha(S_i, C_i, W_i, c_i)$ where $C_i$, $1 \leq i \leq p$ are the connected components of $G[C \setminus \Omega]$, 
$S_i = N_G(C_i)$, $W_i = C_i \cap S_i$ and $c_i$ are classes (still to be guessed). Recall that, by~\cite{BoTo01}, $(S_i, C_i)$ are full blocks.

Intuitively, let $(G[F],X)$  be an optimal partial solution for  $\beta(S,C,\Omega,W,c)$. We denote by $H = (F, W, E_H)$ the $(t+1)$-terminal recursive graph corresponding to $G[F]$ with terminal set $W$, and let $H_i = (F_i, W_i, E_i)$ be its trace on the smaller block $(S_i, C_i)$. Hence $F_i = F \cap (S_i \cup C_i)$, $W_i = W \cap S_i$ and $E_i = E(G[F_i])$. Also denote $X_i = X \cap (S_i \cup C_i)$. Observe that $H$ is obtained from the smaller $H_i$s as follows:
\begin{itemize}
\item on each $H_i$, we introduce the terminals of $W \sm W_i$, obtaining a graph $H_i^+ = (F_i \cup W, W, E^+_i)$ with $W$ as set of terminals and with $E^+_i = E(G[F_i \cup W])$ as edge set.
\item we perform a sequence of joins, gluing one by one $H^+_1, H^+_2,\dots, H^+_p$ on the same set of terminals $W$.
\end{itemize}

Formally, let us first define $\delta_i(S, C, \Omega, W,c_i^+)$ to be the size of the largest partial solution $(G[F^+_i],X^+_i)$ compatible with $(S,C,\Omega,W,c^+_i)$ such that $F^+_i \subseteq \Omega \cup C_i$. (This partial solution was denoted above by $H^+_i$,  $F^+_i$ corresponds to $F_i \cup W$, and $X^+_i$ is $X_i \cup (X \cap W)$.) Consider the composition operation $in(W_i,W)$ which takes two $(t+1)$-terminal graphs, with terminal sets $W_i$ and $W$ respectively, 
and composes them into a new $(t+1)$-terminal graph having $W$ as set of terminals. In the gluing operation, terminal number $j$ of $W_i$ is glued on terminal number $k$ of $W$ if and only if they correspond to the same vertex of $G$. Hence, this composition operation $in(W_i,W)$ only depends on $W_i$ and $W$. Let $X_W \subseteq W$, let $G[W]$ denote the base $(t+1)$- having $W$ as set of terminals, and $c_W$ be the homomorphism class $h(G[W],X_W)$. 

\begin{lemma}
\begin{equation}\label{eq:delta}
	\delta_i(S, C, \Omega, W, c_i^+) = \max_{c_i, c_W\text{~s.t.~} \odot_{in(W_i,W)}(c_i,c_W) = c^+_i} \alpha(S_i,C_i,W_i, c_i) + |term(c_W,W) \setminus term(c_i,W_i)|
\end{equation}
over all classes $c_i$ and $c_W$ such that $\odot_{in(W_i,W)}(c_i,c_W) = c^+_i$ and $c_W = h(G[W],X_W)$ for some $X_W \subseteq W$.
\end{lemma}
\begin{proof}
Let $(G[F^+_i],X^+_i)$ be a maximal partial solution compatible with $(S,C,\Omega,W,c^+_i)$ such that $F^+_i \subseteq \Omega \cup C_i$. 
Denote $F_i = F^+_i \cap (S_i \cup C_i)$, $X_i = X^+_i \cap (S_i \cup C_i)$, $X_W = X \cap W$. Observe that $(G[F_i],X_i)$ is a partial solution compatible with $(S_i,C_i,W_i, c_i)$ for some class $c_i$, that $c_W = h(G[W],X_W)$, and these classes must 
satisfy $\odot_{in(W_i,W)}(c_i,c_W) = c^+_i$. 
Hence $\delta_i(S,\Omega,C,W,c_i^+)$ is at most equal to the right-hand side of the equation (note that $term(c_W,W) \setminus term(c_i,W_i) = X^+_i \setminus X_i$).

Conversely, let $c_i, c_W$ be the classes maximizing the right-hand side of the equation. Take a maximum partial solution $(G[F_i],X_i)$ contained in $S_i \cup C_i$, compatible with $(S_i,C_i,W_i, c_i)$, where $ \odot_{in(W_i,W)}(c_i,c_W) = c^+_i$. Then the graph $(F_i \cup W, W, E(G[F_i \cup W]))$ together with the vertex subset $X_i \cup term(c_W,W)$ is a partial solution compatible with  $(S,C,\Omega,W,c^+_i)$, and the equality follows.
\end{proof}

 We introduce another notation $\gamma_i(S,C,\Omega,W,c)$, corresponding to the largest partial solution compatible with $(S,C,\Omega,W,c)$, contained into $\Omega \cup C_1 \cup \dots \cup C_i$. It corresponds to the gluing of some partial solutions $(H_1^+,X^+_1), \dots (H_i^+,X^+_i)$. 
 
 \begin{lemma}
 Function $\gamma_i$ is computed as follows.
\begin{equation}\label{eq:gamma1} 
	\gamma_1(S,C,\Omega,W,c) = \delta_1(S,\Omega,C,W,c)
	\end{equation}
For any $i$, $2 \leq i \leq p$, 
\begin{equation}\label{eq:gammai}
\gamma_i(S,C,\Omega,W,c) = \max_{c',c''} \gamma_{i-1}(S,C,\Omega,W,c') + \delta_i(S,\Omega,C,W,c'') - |term(c',W)|,
\end{equation}
 over all characteristics $c',c''\in \cC$ such that $\odot_{g(W)}(c',c'') = c$, where $g(W)$ is the composition operation corresponding to a join operation on $W$. I.e., the matrix $m(g(W))$ of $g(W)$ has $|W|$ rows, and $m_{j,1}(g(W))=m_{j,2}(g(W))=j$ for each row $j$.
\end{lemma}
\begin{proof}
The proof is trivial for $\gamma_1$.

Now for any $F \subseteq \Omega \cup C_1 \cup \dots \cup C_i$, note that $(G[F],X)$ is a partial solution compatible with $(S,C,\Omega,W,c)$ if and only if $(G[F \setminus C_i],X \setminus C_i)$ (resp. $(G[F \setminus (C_1 \cup \dots \cup C_{i-1})], X \setminus (C_1 \cup \dots \cup C_{i-1}))$) are partial solutions compatible with  $(S,\Omega,C,W,c'')$ (resp. $(S,C,\Omega,W,c')$) and $\odot_{g(W)}(c',c'') = c$. The term  $|term(c',W)|$ corresponds to $X \cap W$ and avoids over-counting of these vertices.
\end{proof}

The following result is a direct consequence of the definition of $\beta$ and $\gamma$ functions.
\begin{lemma}
\begin{equation}\label{eq:beta}
\beta(S,C,\Omega,W,c) = \gamma_p(S,C,\Omega,W,c).
\end{equation}
\end{lemma}

We claim that computing, for a fixed quadruple $(S,C,\Omega,W)$, the values $\beta(S,C,\Omega,W,c)$ from values $\alpha$,  takes $O(n^2)$ time. Again by~\cite{FKTV08}, the smaller blocks $(S_i,C_i)$ can be listed in $O(m)$ time. For each $i$, the computation of  function $\delta_i(S,\Omega,C,W,c_i^+)$ takes $O(|S_i|+|C_i|)$ time, because we need to access the values $\alpha(S_i,W_i,C_i,c_i)$. The sum of these values is at most $n+m$~\cite{FKTV08}. Computing $\gamma_i(S,C,\Omega,W,c)$ from values $\gamma_{i-1}$ and $\delta_i$ can be done in $O(n)$ time for each $i$.

Therefore the running time of the algorithm is the number of quintuples $(S, C, \Omega, W, c)$ times $n^2$, that is $O(|\Pi_G|n^{t+4})$.

\paragraph{The global solution.} It can be obtained by considering the (special) full block $(\emptyset, V)$. 
\begin{lemma}
The solution size is 
\begin{equation}\label{eq:final}
\max_c \alpha(\emptyset, V, \emptyset, c),
\end{equation}
 over all accepting classes $c$, i.e., classes such that $(h(G,X)=c)$ implies that $\cP(G,X)$. 
\end{lemma}
\begin{proof}
By definition of regular properties and of $\alpha(\emptyset, V, \emptyset, c)$, our problem has a solution of size at least  $\max_c \alpha(\emptyset, V, \emptyset, c)$ over accepting classes $c$.

Let $(G[F],X)$ be a maximum size solution for our problem. By Proposition~\ref{pr:CompTr}, this solution is compatible with $\alpha(\emptyset, V, \emptyset, c)$ for the class $c$ of the $(t+1)$-terminal graph graph $(F, \emptyset, E(G[F]))$, which achieves the proof of the lemma.
\end{proof}

This latter computation takes constant time. 

The total running time of the algorithm is $O(|\Pi_G|n^{t+4})$. Note that, instead of keeping the size of the largest solution $(G[F],X)$, we could explicitly store the vertex subsets $(F,X)$ of $G$.

\subsection{Extensions}\label{ss:extension}

Theorem~\ref{theorem_main} can be extended to \emph{weighted} and \emph{annotated} versions of problem \ois{}, for any $t\geq 0$ and any regular property $\cP$. 

\defproblem{\owais{}}{A graph $G=(V,E)$ a weight function $w : V \rightarrow \mathbb{R}$, a set $U \subseteq V$ of annotated vertices and a number $t$.}{Find sets $X \subseteq F\subseteq V$ such that $F$ contains $U$, the  induced subgraph $G[F]$ is of treewidth at most $t$, property $\cP(G[F],X)$ is true and $X$ is of maximum weight under these conditions.}

\begin{theorem}\label{th:extensionswa}
For any fixed $t$ and any regular property $\cP$, the problem \owais{} is solvable in $|\Pi_G| n^{O(1)}$ time, when $\Pi_G$ is given in the input. 

In particular the problem can be solved in $\cO(\npmc^n)$ time for arbitrary graphs, and in polynomial time for classes of graphs with polynomial number of minimal separators. 
\end{theorem}

For this purpose, we slightly adapt the definitions of $\alpha$ and $\beta$ functions. In order to force the annotated vertices to be in $F$, each value $\alpha(S,C,W,c)$ 
(resp. $\beta(S,C,\Omega, W,c)$ such that $U \cap S \not\subseteq W$ (resp. $U \cap \Omega \not\subseteq W$) is immediately set to $-\infty$, meaning that such a partial solution is rejected.

In order to maximize the weight of the solution, the values  $\alpha(S,C,W,c)$ (respectively $\beta(S,C,\Omega, W,c)$) will correspond to the maximum weight over partial solutions compatible with 
$(S,C,W,c)$ (resp. $(S,C,\Omega, W,c)$). In the algorithm, we simply replace the cardinality of sets (e.g., $|X|$ in Equation~\ref{eq:base_case}, $|term(c',W)|$ is Equation~\ref{eq:gammai} and $|term(c_W,W) \setminus term(c_i,W_i)|$ in 
Equation~\ref{eq:delta}) by the weights of these sets. 

We also point out that the weights can be negative. In particular, we can use Theorem~\ref{th:extensionswa} to compute an induced subgraph $G[F]$ of treewidth at most $t$ and a subset $X \subseteq F$ such that $\cP(G[F],X)$ is true, and $X$ is of minimum size (or weight) under these conditions.

\medskip

One can imagine more extensions of Theorems~\ref{theorem_main} and~\ref{th:extensionswa}. A natural one consists in finding sets $X$ and $F$ such that the size of $X$ is \emph{exactly} an input value $v$. For this purpose, we can adapt our definitions of $\alpha$ and $\beta$ to store, for each possible value $v' \leq v$, a boolean $\alpha(S,C,W,c,v')$ (resp. $\beta(S,C,\Omega,W,c,v')$), set to \emph{true} if and only if there exists partial solution $(G[F'],X')$ compatible with $(S,C,W,c)$ (resp. $(S,C,\Omega,W,c)$ such the size of $X'$ is exactly $v'$. The computation of $\alpha$ and $\beta$ is quite straightforward, by adapting Equations~\ref{eq:base_case} to~\ref{eq:final}. The complexity of the algorithm is multiplied by a polynomial factor. 

Even more involved, we can consider properties $\cP(G,X_1,\dots, X_p, E_1, \dots, E_q)$, where each $X_i$ is a vertex subset and each $E_j$ is an edge subset of graph $G$. The notion of regularity extends in a very natural way to several variables. Recall that Borie \textit{et al.}~\cite{BPT92} proved that \emph{all} properties expressible by CMSO-formulae are regular, so we are allowed to use any (fixed) number of free variables corresponding to vertex sets and edge sets. 

Let $t\geq 0$ be an integer and $\cP(G,X_1,\dots, X_p, E_1, \dots, E_q)$ be a regular property on graphs and vertex subsets $X_i$ and edge subsets $E_j$. We define the following generic problem. 

\defproblem{\cis{}}{A graph $G$, integer values $v_1,\dots, v_p \leq n$ and $w_1,\dots, w_p \leq \frac{n(n-1)}{2}$}
{Find $F\subseteq V$, sets $X_i \subseteq F$ and $E_j \subseteq E(G[F])$ such that
the  induced subgraph $G[F]$ is of treewidth at most $t$, $\cP(G,X_1,\dots, X_p, E_1, \dots, E_q)$ is true, each set $X_i$ is of size $v_i$ and each set $E_j$ is of size $w_j$.}
  
Since property $\cP$ is regular, we need to adapt the definition of partial solutions to more variables (again very naturally) and then we define as above boolean functions  $$\alpha(S,C,W,c,v'_1, \dots, v'_p,w'_1 \dots, w'_q),$$ respectively $$\beta(S,C,\Omega,W,c,v'_1, \dots, v'_p,w'_1 \dots, w'_q)$$ to be \emph{true} if there exists a partial solution $(G[F'],X'_1,\dots, X'_p, E'_1, \dots, E'_q)$ compatible with $(S,C,W,c)$ (resp. $(S,C,\Omega,W,c))$ such that each $X'_i$ is of size $v'_i$ and each $E'_j$ is of size $w'_j$. For computing the $\alpha$ and $\beta$ values, we must again adapt Equations~\ref{eq:base_case} to~\ref{eq:final}. Basically, for each class $c$, the function $term(c,W)$ used in the equations for a homomorphism class $c$ and an order set of terminals $W$ must now return each intersection of type $X'_i \cap W$ for vertex sets and $E'_j \cap G[W]$ for edge sets. These intersections will be used to avoid overcounting when glueing partial solutions. The complexity of the algorithm becomes larger by a factor of  $n^{\cO(p+q)}$.

Therefore we can solve problems like finding, among maximum induced subgraph of treewidth at most $t$, the one with minimum dominating set.

%

%
%


\section{Applications}\label{se:applications}
In this section we discuss several applications of Theorem~\ref{theorem_main}. Our results are summarized in the following theorem. Recall that the problems have been defined in the \emph{Introduction}.

\begin{theorem}\label{thm:applications} Let $G$ be an $n$-vertex graph given together with the set of its potential maximal cliques $\Pi_G$. Then
\begin{itemize}
\item
 \amims,
 \item \mislkc,
 \item \mislf, where $\cF$ contains a planar graph, 
   \item \igp, 
    and 
\item \kig
   
 \end{itemize}
 are solvable in time 
$|\Pi_G|\cdot   n^{\cO(1)}$. Here the hidden constants in $\cO$ depend on $m,p,\ell$, 
$\cF$, $t$, and $\varphi$. 
\end{theorem} 

Combined with Proposition~\ref{pr:listing_pmc_gen}, Theorem~\ref{thm:applications} implies the following. 
\begin{corollary}\label{cor:exact_algorithms}
Let $G$ be an $n$-vertex graph. All problems from Theorem~\ref{thm:applications} are solvable in time  $\cO(\npmc^n)$.
\end{corollary}

The proof of Theorem~\ref{thm:applications} follows from 
  Theorem~\ref{theorem_main} and 
  Lemmata~\ref{lemma:parity_CMSOL},
\ref{lemma:long_cycle_CMSOL},
\ref{lemma:minor_CMSOL},
\ref{lemma:packing_CMSOL},  and
\ref{lemma:kinpath_CMSOL}.

Let us remark that Theorem~\ref{thm:applications} also holds for different modifications of these problems, like requirements of the maximum induced subgraph being connected, of maximum vertex degree at most some constant $\Delta$, etc. Such modifications easily capture problems like   computing a longest induced path, cycle, or an induced tree with given maximum vertex degree.

\medskip
\noindent
\textbf{Hitting and packing cycles of length $0\imod{m}$.} We will need the following  result of Thomassen.

\begin{proposition}[\cite{Thomassen88}]\label{prop:Thomas_tw}
For every integers $\ell,m>0$ there exists an integer $k(\ell,m)>0$ such that the treewidth of a graph with at most $\ell$ vertex-disjoint cycles from $\cF_m$ is at most  $k(\ell,m)$.
\end{proposition}

With the help of Proposition~\ref{prop:Thomas_tw}, we obtain the following lemma.

\begin{lemma}\label{lemma:parity_CMSOL}
\amims{} is a special case of \ois{} with $t=f(\ell,m)$, where  $f$ depends only on  $m$ and $\ell$.
\end{lemma}
\begin{proof} 
For a graph $G$ let $F$ be the maximum vertex set such that $ G[F]$ has at most $\ell$ vertex-disjoint cycles from  ${\cal F}_m$.
We put $f(\ell,m)=k(\ell,m)$, where  $k(\ell,m)$ is the integer from Proposition~\ref{prop:Thomas_tw}. By Proposition~\ref{prop:Thomas_tw}, the treewidth of $G_F$ is at most $f(\ell,m)$.

Then \amims{}  is to maximize $|X|$ for the following property
\[\cP(G[F],X)= \{ F=X \text{ and } G[F] \text{
contains at most $\ell$ vertex-disjoint cycles from }  {\cal F}_m. \}
\]
To show that $\cP(G[F],X)$ is regular, 
we observe that it  is expressible by a CMSO-formula. 
Indeed, this formula expresses that for every partition of $V(G_F)$ into $\ell +1$ subsets, there is a subset containing no cycle from  ${\cal F}_m$.
\end{proof}

%
%
%
%
%


\medskip
\noindent
\textbf{Hitting long cycles.} 
%
%
%
%
%
We   need the following result, which  is due to 
Birmel{\'e} et al. 
\begin{proposition}[\cite{Birmele07}]\label{prop_Birmele}
Graphs without $\ell$ disjoint cycles of length at least $p$  are of treewidth $\cO(\ell^2 p)$.
\end{proposition}

By making use of Proposition~\ref{prop_Birmele}, it is easy to prove the following lemma.

\begin{lemma}\label{lemma:long_cycle_CMSOL}
\mislkc{} is a special case of  \ois{}  with $t= \cO(\ell^2 p)$.
\end{lemma}
\begin{proof} 
For a graph $G$ let $F$ be the maximum vertex set such that $  G[F]$ has at most $\ell$ vertex-disjoint cycles of length at least $p$.
  By Proposition~\ref{prop_Birmele}, the treewidth of $G[F]$ is at most $\cO(\ell^2 p)$. Then we are  maximizing $|X|$ 
for  the following property
\[\cP(G[F],X)= \{ F=X \text{ and } G[F] \text{
contains   $\leq \ell$ vertex-disjoint cycles of length   }  \geq p. \}
\]
To show that this property is regular, we observe that property of not having a cycle of length at least $p$ is expressible in CMSO. Indeed, a property of a set $C$ of vertices to induce a cycle is CMSO, and because $p$ is fixed, the formula expressing the sentence that for every subset $C$ inducing a cycle, the number of elements is at most $p$, is of constant length. 
Because $\ell$ is also fixed,   it is possible to express by a constant size CMSO-formula the sentence that  for every partition in   $\ell + 1$ subsets there is a subset inducing a subgraph without a cycle of length at least $p$.
 \end{proof}

%

\medskip
\noindent
\textbf{Excluding planar minors.}
 The following proposition follows almost directly from the excluded grid theorem  of Robertson and Seymour \cite{RobertsonS-V}, see also \cite{RobSeymT94}.


\begin{proposition}[\cite{RobertsonS-V}]\label{prop:R_S}
For every integer $\ell>0$ and family $\cF$ containing a planar graph, there exists an integer $k(\ell,\cF)>0$ such that the treewidth of a graph with at most $\ell$ vertex-disjoint minor models from $\cF$ is at most  $k(\ell,\cF)$.
\end{proposition}


\noindent

\begin{lemma}\label{lemma:minor_CMSOL} If $\cF$ contains a planar graph, then 
\mislf{} is a special case of  \ois{}  with $t= k(\ell,\cF)$.
\end{lemma}
\begin{proof} 
For a graph $G$ let $F$ be the maximum vertex set such that $ G[F]$ has at most $\ell$ vertex-disjoint models of minors from  $\cF$.
  By Proposition~\ref{prop:R_S}, the treewidth of $G[F]$ is at most  $k(\ell,\cF)$. The property that a graph does not contain a fixed graph as a minor is known to be expressible in CMSO.  This implies that the property
  \[\cP(G[F],X)= \{  F=X \text{ and } G[F] \text{
has   $\leq \ell$ vertex-disjoint minor models   from
   } {\cF}  \}
\]
  is regular.
  \end{proof}

\medskip
\noindent\textbf{Independent packing.}
\begin{lemma}\label{lemma:packing_CMSOL} 
\igp{} is a special case of \ois.
\end{lemma}
\begin{proof} 
 
For a graph $G$ let $F$ be a vertex set such that $G_F=G[F]$ has the maximum number of connected components, and each of the components is in   $\mathcal{G}(t,\varphi)$.
Because the treewidth of every component does not exceed $t$, the treewidth of $G[F]$ does not exceed $t$. We use $cc(G[F])$ to denote the set of connected components of $G[F]$.
Because the property that every connected component has regular is also regular, we have  that the following property is regular
  \[\cP(G[F],X)= \{  [X\subseteq V(G_F)] \wedge [\forall C\in cc(G_F) (C\in \mathcal{G}(t,\varphi)  \wedge   |X\cap C|=1)] .\}
\]
\end{proof}

%
%
%
%
%
%
%
%
%
%
%
%
%
%
%

\medskip
\noindent
\textbf{$k$-in-a-graph.} Because in  \kig{},  $k${} is  part of the input we need the annotated variant of the main theorem (Theorem~\ref{th:extensionswa}).
The following lemma follows from the definition of the problems.

\begin{lemma}\label{lemma:kinpath_CMSOL} 
\kig{} is a special case of  \owais.
\end{lemma}



\section{Graph classes}\label{sec:graph_classes}
In this section we discuss the consequences of Theorem~\ref{thm:applications} for special graph classes. In particular,    by Proposition~\ref{pr:listing_pmc}, every class of graphs with polynomially many minimal separators also has polynomially many potential maximal cliques.  
For example, every $n$-vertex  \emph{weakly chordal} graph, i.e. graph  with no induced cycle or its complement of length greater than four, 
has  $\cO(n^2)$ minimal separators~\cite{BoTo01}.  This class of graphs is a generalization of many graph classes intensively studied in the literature like chordal, split, and interval graphs.
 Another class of graphs of this type is the class of 
 \emph{circular-arc}  graphs,   intersection graphs of a set of arcs on the circle. 
Every circular-arc with  $n$ vertices has  at most $2n^2 -3n$ minimal separators \cite{KloksKW98}. 
The class of  {$d$-trapezoid} graphs is defined as follows.
Let $L_1, \dots, L_d$ be $d$ parallel lines in the plane. A $d$-trapezoid is the polygon obtained by choosing an interval $I_i$ on every line $L_i$ and connecting the left, respectively, right endpoint of $I_i$ with the left, respectively,  right endpoint of $I_{i+1}$.
A graph is a \emph{$d$-trapezoid graph} if it has an intersection model consisting of $d$-trapezoids between $d$ parallel lines. Every
$d$-trapezoid graph has   at most  $(2n-3)^{d-1}$ minimal separators \cite{Kr96}, see also \cite{brandstadt1999graph}. An intersection graph of polygons
  enclosed by a bounding circle is is know as a \emph{polygon-circle graph}.
As  it was observed by Suchan in \cite{Suchan2003}, every 
polygon-circle with $n$ vertices has $\cO(n^2)$ minimal separators.  
See  Fig~\ref{fig:graph_classes} of the \emph{Introduction} for the relations between most known classes of graphs with polynomially many minimal separators. We refer to the encyclopedia of graph classes  \cite{brandstadt1999graph} for definitions of different  graphs  from  Fig~\ref{fig:graph_classes}.

Let us remark that the only information for our algorithms we need is the bound on the number of minimal separators in the specific graph class. While many of the algorithms from the literature for intersection classes of graphs strongly use the intersection model  this is not necessary for our algorithms---they produce correct output regardless of whether the input actually belongs to the specific class of graphs. 
If the number of minimal separators and thus \pmc s is bounded, our algorithm correctly solves the problem. Otherwise, the algorithm correctly reports that the given input is not from the restricted domain. 
Such type of algorithms were called \emph{robust} by Raghavan and Spinrad \cite{RaghavanS03}. For example, while recognition of $d$-trapezoid and polygon-circle graphs is NP-complete \cite{Yannakakis82,Pergel:2007fk}, our algorithm either correctly solves the problem or outputs that the input graph is not $d$-trapezoid or polygon-circle.

\begin{corollary}\label{cor:grapht_classes}
 All problems from Theorem~\ref{thm:applications} are solvable in polynomial time  on classes of graphs from Fig~\ref{fig:graph_classes}. 
\end{corollary}

On several classes of graphs even more general problems can be solved. The observation here is that for many classes of graphs from Fig~\ref{fig:graph_classes}, the treewidth of a graph is upper bounded by some function of other  parameters like the maximum clique-size  or  maximum degree. 

For example, Yannakakis and Gavril 
 \cite{YannakakisG87} have shown that for every fixed $\chi$, a maximum induced subgraph of a chordal graph colorable in $\chi$ colors can be found in polynomial time.  To see why this result follows as a corollary of our theorem, let us observe that  for chordal graphs, as for  all perfect graphs, the chromatic number is equal  to the maximum clique size, see e.g.  \cite{Golumbic80}. On the other hand, 
the treewidth of a chordal graph is known to be equal to the maximum clique size minus one. Thus every induced $\chi$-colorable subgraph of a chordal graph is of treewidth at most $\chi -1$.  Since colorability in a constant number of colors is expressible in CMSO, the result follows.

For other variant of colorings, we need the  the following proposition due to Gaspers et al.
\begin{proposition}[\cite{GaspersKLT09}]\label{prop:degree_tw}
Let $G$ be a graph of maximum vertex degree at most $D$. Then the treewidth of $G$ is at most
\begin{itemize}
\item 
$4D$, if $G$ is a circle graph,
\item $2D$,  if $G$ is a 
weakly chordal  graph or a  circular-arc  graph.
\end{itemize}
\end{proposition}
 Combined with Proposition~\ref{prop:degree_tw}, Theorem~\ref{thm:applications} allows us to show that on several graph classes,  in addition to problems encompassed by Corollary~\ref{cor:grapht_classes}, 
 even larger class of problems can be solved efficiently.  
 For example, \emph{edge coloring} of a graph is an assignment of colors to the edges of the graph so that no two adjacent edges have the same color. The \emph{chromatic index} of a graph is the minimum number of colors required for edge coloring. By Vizing's theorem, for every graph with maximum vertex degree $D$, its chromatic index is either  $D$ or $D+1$. Since   edge coloring in a constant number of colors is expressible in CMSO, we conclude that the problem of finding a maximum induced  edge-colorable in $k$ colors subgraph (for a fixed constant $k$) is solvable in polynomial time on  circle, weakly chordal  and   circular-arc graphs.
%
 Similarly, the problems like for a fixed constant $k$ finding a maximum induced (connected) subgraph of maximum vertex degree at most $k$ are also solvable in polynomial time on these classes of graphs.



The next lemma provides a different set of applications of the main theorem for special graph classes. 

\begin{lemma}\label{lemma:bounds_tw}
Let $G$ be a graph excluding some fixed graph $H$ as a minor.
  Then the treewidth of $G$ is at most
\begin{itemize}
\item 
$f(H)$ for some function $f$ of $H$ only,  if $G$ is a weakly-chordal  graph, and 
\item $3|V(H)|$  if $G$ is a circular-arc  graph.
\end{itemize}
\end{lemma}
\begin{proof}
Let $G$ be a weakly chordal graph excluding $H$ as a minor.
By a theorem from \cite{Fomin:2011fk}, there is a constant $c_H$ such that every $H$-minor-free graph of treewidth at least $c_H k^2$ can be transformed by making only edge contractions  either to a  planar triangulation $\Gamma_k$ of a $(k
\times k)$-grid, or to $\Pi_k$, which is a graph obtained from $G_k$ by adding a universal vertex. Since both $\Gamma_k$ and $\Pi_k$ for $k\geq 3$ contain an induced cycle of length at least $6$, we conclude that the treewidth of $G$  does not exceed some constant depending only on $H$. Indeed, otherwise a contraction of $G$, and hence $G$ too, would contain an induced cycle of length more than $4$.

 \medskip
 
For circular-arc graphs, we can prove the statement of the lemma by using the observation from \cite{KloksKW98} that every potential maximal clique of a circular-arc graph is 
  the union of at most of  three cliques. 
  Thus every circular-arc graph of treewidth at least $3|V(H)|$ should contain a \pmc{} of size at least $3|V(H)|$, and hence a clique of size at least $ |V(H)|$.  Thus every   circular-arc graph of treewidth at least $3|V(H)|$ contains $H$ as a minor.
  \end{proof}

By combining Lemma~\ref{lemma:bounds_tw} with  Theorem~\ref{thm:applications}, we obtain that 
\fd{} is solvable in polynomial time on circular-arc and weakly chordal graphs for every finite family $\cF$ of graphs. The requirement that $\cF$   contains a planar graph can be omitted in this case.

\section{Conclusion}\label{sec:conclusion}
While regular properties and CMSO capture many interesting problems, it seems that the approach based on minimal triangulations is not restricted by these settings. 
Take for example the following problem.
\smallskip

\defproblem{\midcls{}}{A graph $G$,  and a  collection   $\{T_1, T_2, \dots, T_p\}$ of terminal vertices, $T_i\subseteq V(G)$, of size at most $\ell$.}{Find a  set $F \subseteq V(G)$ of minimum size such that  
 $G[F]$ has connected components $C_1, C_2, \dots, C_p$ and for every $1\leq i\leq p$, $T_i \subseteq C_i$.}

\smallskip 

This problem is a generalization of the 
 \textsc{Induced Disjoint Paths}, where for a given set of $p$ pairs of terminals $x_i,y_i$, $1\leq i \leq p$, the task is to find a set of paths connecting terminals such that the vertices from different paths are not adjacent.  
Belmonte et al. \cite{BelmonteGHHKP11} have shown that \textsc{Induced Disjoint Paths} is solvable in polynomial time on chordal graphs.  Because $p$ is part of the input and not fixed, this problem cannot be expressed by a CMSO-formula of constant size. On the other hand, by applying a modification of the dynamic programming algorithm over \pmc s and minimal separators, it is possible to show that this problem is solvable in time  proportional to the number of \pmc s, up to polynomial factor $n^{t+\cO(1)}$.
 
Another example can be the following problem. Let $t$ be an integer. 

\smallskip

\defproblem{\hftg{}}{Graph $G$ and $H$}{Find a  set $F \subseteq V(G)$ of maximum size such that the treewidth of $G[F]$ is at most $t$ and there is a homomorphism from   
 $G[F]$ to $H$.}

\smallskip 
By the classical result of  Yannakakis and Gavril 
 \cite{YannakakisG87},  for every fixed $\chi$, a maximum induced subgraph of a chordal graph colorable in $\chi$ colors can be found in polynomial time. Because coloring into $\chi$ colors is homomorphism in a complete graph on $\chi$ vertices, and because  the treewidth of a $\chi$-colorable chordal graph is at most $\chi -1$, 
 \hftg{} extends this problem. However, the property of having a homomorphism to $H$ is not CMSO-expressible because $H$ is part of the input. Moreover,  it is easy to see that already very special case of graph homomorphism problem, where we are asked for a homomorphism from  a clique 
 of size $k$ (and thus of treewidth $k-1$) to  $H$ is equivalent to deciding if $H$ has a clique of size at least $k$, which is W[1]-hard. Thus homomorphism from $G$ to $H$ parameterized by the treewidth of $G$  is W[1]-hard.  But on the other hand,  dynamic programming  over \pmc s and minimal separators shows that  \hftg{}  is solvable in time proportional to the number of \pmc s, up to polynomial factor $n^{O(t)}$.
 
 Both examples indicate that even more general framework capturing problems solvable in time proportional to   the number of \pmc s can exist. Defining such a  general framework is an interesting open question. 
 
 Another open question concerns counting problems. 
  Our approach does not work for counting problems  due to potential  double counting in the process of computing functions $\alpha $ and $\beta$. We do not exclude a possibility that with additional (clever) ideas the main algorithm  of the paper can also count maximum sets with regular properties but we do not know how to do it, and leave it as an interesting open question. 

Another problem which   seems to be very much related but still cannot be handled directly by our approach is \textsc{Connected Feedback Vertex Set}, where we are asked to find a minimum feedback vertex set inducing a connected subgraph.  Interestingly, out approach works without problems for  \textsc{Maximum Induced Tree}, where the task is to find a minimum feedback vertex set such the remaining graph is connected, i.e. a tree. 

\medskip\textbf{Acknowledgements} We thank Bruno Courcelle, Daniel Lokshtanov, Mamadou  Kant\'e, Dieter Kratsch, Saket Saurabh, Bich Dao and Dimitrios M. Thilikos for fruitful discussions and useful suggestions on the topic of the paper. 

\newpage

\bibliographystyle{siam}
\bibliography{MITPP.bib}
 
\end{document}